\def\year{2022}\relax
\documentclass[letterpaper]{article} % DO NOT CHANGE THIS
\usepackage{aaai22}  % DO NOT CHANGE THIS
\usepackage{times}  % DO NOT CHANGE THIS
\usepackage{helvet}  % DO NOT CHANGE THIS
\usepackage{courier}  % DO NOT CHANGE THIS
\usepackage[hyphens]{url}  % DO NOT CHANGE THIS
\usepackage{graphicx} % DO NOT CHANGE THIS
\def\UrlFont{\rm}  % DO NOT CHANGE THIS
\usepackage{natbib}  % DO NOT CHANGE THIS AND DO NOT ADD ANY OPTIONS TO IT
\usepackage{caption} % DO NOT CHANGE THIS AND DO NOT ADD ANY OPTIONS TO IT
\DeclareCaptionStyle{ruled}{labelfont=normalfont,labelsep=colon,strut=off} % DO NOT CHANGE THIS
\usepackage{algorithm}
\usepackage{algorithmic}

\usepackage{newfloat}
\usepackage{listings}
\floatstyle{ruled}
\newfloat{listing}{tb}{lst}{}
\floatname{listing}{Listing}
\title{Complexity of Deliberative Coalition Formation}
\author{
    Edith Elkind, Abheek Ghosh\thanks{Supported by Clarendon Fund and SKP Scholarship.}, Paul Goldberg
}
\affiliations{
    Department of Computer Science, University of Oxford, UK.\\
    \{edith.elkind,abheek.ghosh,paul.goldberg\}@cs.ox.ac.uk
}

\graphicspath{ {./images/} }
\usepackage{amsmath,amsthm,amsfonts,amssymb,bm,bbm,mathtools}
\usepackage{xcolor}

\newtheorem{theorem}{Theorem}
\newtheorem{proposition}[theorem]{Proposition}
\newtheorem{lemma}[theorem]{Lemma}
\newtheorem{corollary}[theorem]{Corollary}

\theoremstyle{definition}
\newtheorem{definition}{Definition}[section]

\DeclareMathOperator{\bb}{\mathbf{b}}
\DeclareMathOperator{\bd}{\mathbf{d}}

\DeclareMathOperator{\bD}{\mathbf{D}}

\DeclareMathOperator{\bR}{\mathbb{R}}

\DeclareMathOperator{\cS}{\mathcal{S}}
\DeclareMathOperator{\cX}{\mathcal{X}}

\DeclareMathOperator{\cV}{\mathcal{V}}

\DeclarePairedDelimiter\ceil{\lceil}{\rceil}
\DeclarePairedDelimiter\floor{\lfloor}{\rfloor}

\begin{document}

\maketitle

\begin{abstract}
\citet{elkind2021aaai,elkind2021arxiv} introduced a model for deliberative coalition formation, where a community wishes to identify a strongly supported proposal from a space of alternatives, in order to change the status quo. 
In their model, agents and proposals are points in a metric space, agents' preferences are determined by distances, and agents deliberate by dynamically forming coalitions around proposals that they prefer over the status quo. The deliberation process operates via \textit{$k$-compromise} transitions, where agents from $k$ (current) coalitions come together to form a larger coalition in order to support a (perhaps new) proposal, possibly leaving behind some of the dissenting agents from their old coalitions. A deliberation succeeds if it terminates by identifying a proposal with the largest possible support. For deliberation in $d$ dimensions, Elkind et al.~consider two variants of their model: in the {\em Euclidean model}, proposals and agent locations are points in ${\mathbb R}^d$
and the distance is measured according to $||\cdot||_2$; and in the {\em hypercube model}, proposals and agent locations are vertices 
of the $d$-dimensional hypercube and the metric is the Hamming distance.
They show that in the continuous model $2$-compromises are guaranteed to succeed, but in the discrete model for deliberation to succeed it may be necessary to use $k$-compromises with $k\ge d$.
We complement their analysis by 
(1) proving that in both models it is hard to find a proposal 
with a high degree of support, and even a $2$-compromise transition 
may be hard to compute;
(2) showing that a sequence of $2$-compromise transitions 
may be exponentially long;
(3) strengthening the lower bound on the size of the compromise  
for the $d$-hypercube model from $d$ to $2^{\Omega(d)}$.
\end{abstract}

\noindent 

\section{Introduction}\label{sec:intro}
Imagine a nominally democratic country where the current ruling party has been in power for many years, through a combination of clever political strategizing and a range of more or less non-democratic means (e.g., gerrymandering, vote suppression, media control, etc.). Even if, initially, political parties and alliances other than the ruling party may differ in positions they take on various issues facing the society, after many years of being out of power, they may choose to focus on what unites them rather than on what differentiates them, and seek to identify a common platform that would enjoy broad popular support and enable them to overturn the status quo. In doing so, they may pursue a variety of goals: to increase their chances of winning seats in an election, to place themselves in a better position in post-electoral bargaining, or---in the most extreme cases---to initiate a revolt against the current regime. This common platform may be quite different from the party's true position---what matters is that it attracts a large number of supporters and that it offers an improvement over the status quo, i.e., is preferable to the policies of the current governing party.

In a recent paper, \citet{elkind2021aaai,elkind2021arxiv} propose
a simple model that aims to capture the essential features of such scenarios\footnote{The conference version of their paper~\cite{elkind2021aaai} was published in AAAI'21, and an extended version was presented at COMSOC'21 and is available from arXiv~\cite{elkind2021arxiv}. One of the models we consider, namely, the hypercube model, is only described in the arXiv version, so in what follows we refer to the arXiv version only.}. 
In this model, both voters and proposals are associated with points in a metric space, with voters' preference being driven by distances: voters prefer proposals that are close to them to ones that are further away. The number of voters is finite, but the set of feasible proposals may be any (potentially infinite) subset of the metric space.
There is also a distinguished point, referred to as the {\em status quo} and denoted by $r$. A voter $v$ approves a proposal $p$ if her distance to $p$ is strictly less than her distance to $r$. Voters deliberate in order to identify a proposal that is supported by as many voters as possible. At each point, each voter 
selects some approved proposal to support, with voters who support a given proposal forming a {\em deliberative coalition} around it. This collection of deliberative coalitions---a deliberative coalitions structure---evolves based on {\em transition rules}: for instance, one transition rule allows two coalitions to identify a new proposal supported by all members of both coalitions and to form a new joint coalition around it. The transition rules aim to capture the behavior of agents who are consensus-driven---i.e., they desire to form a large coalition to overturn the status quo---and myopic, in the sense that they make a decision whether to participate in a transition based on the outcome of that transition only and not the entire deliberative process.

In their work, \citet{elkind2021arxiv} primarily focus on the power of various transition rules to enable the identification of popular proposals, in a range of metric spaces. They do not investigate the complexity of the associated algorithmic challenges, and only offer very crude bounds on the number of steps it may take the deliberation process to converge. 

\subsection{Our Contribution}
Our goal is to complement the analysis of \citet{elkind2021arxiv},
by exploring the complexity of the deliberation process in their model. We focus
on two deliberation spaces: in the {\em hypercube} model, voters and proposals are 
vertices of the $d$-dimensional hypercube, with distances measured according to the Hamming distance, and in the {\em Euclidean} model voters and proposals are
elements of the $d$-dimensional Euclidean space, with $||\cdot||_2$ being the underlying distance measure. In both models, each point of the underlying metric space is considered to be a feasible proposal.
We consider three types of questions:
\begin{itemize}
    \item What is the computational complexity of identifying a proposal approved by as many voters as possible, both from a centralized perspective (i.e., how can we compute a popular proposal given the positions of all voters), and from a decentralized perspective (i.e., how can a group of voters identify the next step in the deliberative process)? We consider both the worst-case complexity of this problem and its parameterized complexity, with two natural parameters being the number of voters and the dimension of the space.
    \item How many transitions may be necessary for a deliberation to converge? \citet{elkind2021arxiv} show that deliberation always converges after at most $n^n$ steps, where $n$ is the number of voters; we improve this upper bound to $2^n$ and derive exponential lower bounds for both of the models we consider.
    \item How many coalitions need to be involved in each deliberation step to ensure that a most approved proposal is identified? \citet{elkind2021arxiv} prove that in Euclidean deliberation spaces $2$-coalition deals are sufficient irrespective of dimension, and in a $d$-dimensional hypercube we may need transitions involving at least $d$ coalitions; we improve this lower bound from $d$ to $2^{\Theta(d)}$. 
\end{itemize}

% Our aim is to understand whether the agents can succeed at identifying credible alternatives to the status quo if they conduct deliberation in a certain way. We concentrate on how the properties of the underlying abstract space of proposals and the coalition formation operators available to the agents affect our ability to provide guarantees on the success of the deliberative coalition formation process. We show that, as the complexity of the proposal space increases, more sophisticated forms of coalition formation are required in order to assure success. Intuitively, this seems to suggest that complex deliberative spaces require more sophisticated coalition formation abilities on the side of the agents.

The work of \citet{elkind2021arxiv} is an important step towards modeling coalition formation in the presence of a status quo option. Such a theory provides formal foundations for the design and development of practical systems that can support successful deliberation, for instance, by helping agents to identify mutually beneficial compromise positions (cf. \cite{di2016system}). Our paper supplements the analysis of
\citet{elkind2021arxiv} by resolving several challenging open questions posed by their work.

% If the agents and the proposals lie in a dense Euclidean set (Euclidean model), they showed that there exists a deliberation sequence with $2$-compromise transitions that succeeds after a polynomial number of transitions. While, if the agents and the proposals have approval support over a finite set of agendas of size $d$ (approval model), then $k$-compromises ($k > 2$) may be necessary; they give a linear lower bound ($k \ge d$) and exponential upper bound for $k = 2^{O(d)}$.

% Our results are three fold, we show that for a successful deliberation the following may happen: (1) agents from a very large number of coalitions have to simultaneously join forces to form a larger coalition ($k = 2^{\Omega(d)}$ for approval model); (2) the sequence of compromise transitions may be exponentially long (for both Euclidean and approval models); (3) even for sequence of compromise transitions that are polynomial in length, it may be NP-hard to make some of the transitions in the sequence (for both models).

% \citet{elkind2020deliberative,elkind2021united} introduced a model for deliberative coalition formation, where a community wishes to identify a strongly supported proposal from a space of alternatives, in order to change the status quo. They describe a deliberation process in which agents dynamically form coalitions around proposals that they prefer over the status quo. The dynamic deliberation process works through $k$-compromise transitions. 

\subsection{Related Work}
Our work builds directly on the work of \citet{elkind2021arxiv}. In turn, their paper belongs to a rich tradition in political science that studies {\em spatial coalition formation}~\cite{coombs1964theory,enelow1984spatial,merrill1999unified,vries1999governing}. An important feature of their model that sets it somewhat apart from prior work is the presence of a special reference point, i.e., the status quo. The agenda of social choice in the presence of status quo has recently been pursued by Shapiro, Talmon, and co-authors in a series of papers~\cite{shapiro2018incorporating,shahafST19,bulteau2021aggregation,abramowitz2020amend}.

The study of group deliberation is a broad and interdisciplinary area, see, e.g.,
\cite{austen2005deliberation,hafer2007deliberation,patty2008arguments,list2011group,list2013deliberation,rad2021deliberation,perote2015model,karanikolas2019voting}; we refer the reader to the work
of \citet{elkind2021arxiv} for further discussion. In particular, an important consideration is whether simple deliberation protocols that only involve a small number of participants can achieve a desirable outcome~\cite{goel2016towards,fain2017sequential};
this question is similar to the one considered in Section~\ref{sec:beyond} of our paper.

The process of deliberation that we consider can be viewed as one of dynamic
coalition formation~\cite{konishi2003coalition,arnold2002dynamic,chalkiadakis2012sequentially}. In particular, the transitions from one deliberative coalition structure to the next one can be viewed as a form of better response dynamics, where agents prefer deliberative coalitions whose proposal they approve to those whose proposal they do not approve,
and, among coalitions whose proposal they approve, they favor larger coalitions.
Potential functions are commonly used in the literature to establish upper bounds on convergence time of better response dynamics (see, e.g., \cite{tardosW07}); however, our
use of a potential function to establish a lower bound (Section~\ref{sec:length}) is somewhat unconventional.

\section{Preliminaries}\label{sec:prelim}
Following \citet{elkind2021arxiv}, we define a {\em deliberation space}
as a $4$-tuple $(\cX,\cV,r,\rho)$, where 
$\cX$ is a (possibly infinite) set of {\em proposals}, 
$\cV = \{v_1, \ldots, v_n\}$ is a set of of $n$ {\em agents}, 
$r\in\cX$ is a special element of $\cX$, which we will refer to as 
the {\em status quo}, and $\rho$ is a metric on $\cX\cup\cV$.
An agent $v_i$ is said to {\em approve} a proposal
$x\in\cX\setminus\{r\}$ if $\rho(v_i, x)< \rho(v_i, r)$: intuitively, 
$v_i$ approves $x$ if $x$ is more representative of her position than $r$ is.
It will be convenient to require that each agent approves at least one proposal, 
i.e., to exclude agents that are happy with the status quo.

We will consider two families of $d$-dimensional deliberation spaces,
where $d$ is a positive integer:
Euclidean deliberation spaces and hypercube deliberation spaces.
We refer to \citet{elkind2021arxiv} 
for a discussion of some other deliberation spaces.
\begin{itemize}
    \item $d$-\textbf{Euclidean} The space of proposals $\cX$ is  $\bR^d$,
    $r=(0, 0, \dots, 0)$, and $v_i\in \bR^d$ for each $i\in [n]$.
    The metric $\rho$ is the usual Euclidean norm: $\rho(x,y) = ||x-y||_2$.
    \item $d$-\textbf{Hypercube} The space of proposals $\cX$ is  $\{0,1\}^d$, 
    $r=(0, 0, \dots, 0)$, and $v_i\in\{0, 1\}^d$ for each $i\in [n]$.
    The metric $\rho$ is the Hamming distance: $\rho(x,y) = ||x-y||_1$.
\end{itemize}

At any point in the deliberation process, the set of agents is split
into deliberative coalitions. A {\em deliberative coalition} is 
a pair $\bd=(C, x)$, where $C$ is a non-empty subset of $\cV$, $x\in\cX\setminus\{r\}$, 
and each agent in $C$ approves $x$. 
When convenient, we identify a deliberative coalition 
$\bd = (C,x)$ with its set of agents $C$, and say that agents in $C$
{\em support} $x$. 
%For technical reasons we allow empty deliberative coalitions, i.e., coalitions $(C,x)$ with $C=\varnothing$.

A {\em deliberative coalition structure}
is a set $\bD = \{\bd_1, \ldots, \bd_m\}$, $m \ge 1$, such that:
\begin{itemize}
    \item $\bd_j = (C_j, x_j)$ is a deliberative coalition for each $j \in [m]$;
    \item $\cup_{j \in [m]} C_j = \cV$ 
    and $C_j \cap C_\ell = \varnothing$ for all $j, \ell \in [m]$ such that $j \neq \ell$.
\end{itemize}

The agents start out in some deliberative coalition structure, and then
this structure evolves according to {\em transition rules}. 
\citet{elkind2021arxiv}~define several such rules; in this paper, we focus
on {\em $k$-compromise transitions}. The intuition behind all these
rules is that agents seek to form large coalitions and act myopically, 
i.e., an agent in a deliberative coalition $(C, x)$ 
is willing to deviate so as to end up in a coalition 
$(C', x')$ if (i) she approves $x'$ and (ii) $|C'| > |C|$.
Specifically, under a $k$-compromise transition, 
agents from $\ell\le k$ existing deliberative coalitions 
$\bd_1, \dots, \bd_\ell$, where $\bd_j=(C_j, x_j)$ for $j\in [\ell]$, 
get together and identify a proposal $x$ such that $t$ of them approve
$x$ and $t>|C_j|$ for each $j\in [\ell]$. Then all agents
who approve $x$ form a deliberative coalition that supports proposal $x$.
The agents in $C_j$ who do not approve $x$ stay put, i.e., 
they end up in deliberative coalition $(C'_j, x_j)$ with $C'_j =  C_j\setminus C$
(note that $C'_j$ may be empty).
This notion is formalized as follows.

\begin{definition}[$k$-Compromise Transitions]
A pair of coalition structures $(\bD,\bD')$ forms a {\em $k$-compromise transition} if there exist $\ell$ coalitions $\bd_1, \ldots, \bd_\ell \in \bD$, $1\le \ell\le k$, 
where $\bd_j=(C_j, x_j)$ for $j\in[\ell]$, such that $\bD'$ 
is obtained from $\bD$ by (1) removing $\bd_1, \dots, \bd_\ell$, 
(2) adding a deliberative coalition $(C, x)$ such that
$C\subseteq \cup_{j\in [\ell]}C_j$ and for each $j\in [\ell]$ 
it holds that $|C|>|C_j|$ and $C\cap C_j$ consists 
of all agents in $C_j$ that approve $x$;
(3) for each $j\in [\ell]$ such that $C_j\not\subseteq C$ adding 
a deliberative coalition $(C_j\setminus C, x_j)$. 
\end{definition}

We say that a deliberative coalition structure $\bD$ is {\em $k$-terminal} if
there does not exist a $k$-compromise transition of the form $(\bD, \bD')$.
A {\em $k$-deliberation} is a sequence of $k$-compromise transitions such that
for the last transition $(\bD, \bD')$ in this sequence it holds that $\bD'$
is terminal. We will refer to $\bD'$ as the {\em outcome} of the respective $k$-deliberation.

We define the {\em score} of a proposal $x\in\cX\setminus\{r\}$ 
as the number of agents in $\cV$ who approve $x$. We say that 
a proposal $x\in \cX\setminus\{r\}$ is {\em popular} if its score is at least as high
as that of any other proposal in $\cX\setminus\{r\}$.
A deliberative coalition structure $\bD$ is {\em successful}
if it contains a deliberative coalition $(C, x)$ such that $x$
is a popular proposal and $C$ consists of all agents who approve $x$.
A $k$-deliberation is {\em successful} if its outcome is successful.

Note that if there is a majority-approved proposal, a successful 
$k$-deliberation identifies some such proposal, enabling 
a majority-supported change to the status quo.

% \begin{example}
% Consider a $2$-Euclidean deliberation space with $5$ agents: $v_1 = (1,0), v_2 = (1,1), v_3 = (0,1), v_4 = (-1,1), v_5 = (-1,0)$. The deliberative coalition structure $\bD = ((\{v_1, v_2, v_3\},(1/2,1/2)),(\{v_4, v_5\},(-1,0)))$ is not terminal w.r.t. $2$-compromise transitions because agent $v_4$ can make a compromise with $\{v_1, v_2, v_3\}$ to form $\bD' = ((\{v_1, v_2, v_3, v_4\},(1/5,2/5)),(\{v_5\},(-1,0)))$. This also shows that $\bD$ is not successful. But, it can be checked that $\bD'$ is successful.
% \end{example}

An important result of \citet{elkind2021arxiv} is that a deliberation process with $k$-compromise transitions always terminates.
\begin{theorem}\cite{elkind2021arxiv}\label{thm:compromisenn}
For each integer $k$
with $2\le k\le n$ a $k$-deliberation can have at most $n^n$ transitions.
\end{theorem}
This result holds for any deliberation space, 
so in particular both for the $d$-Euclidean space and the $d$-hypercube
for any $d\ge 1$.

% All omitted proofs are provided in the Appendix.
\section{Complexity of Finding Popular Proposals}\label{sec:complex}
In this section, we focus on the complexity-theoretic challenges
presented by the deliberation process. We first consider 
two computational problems: {\sc Score} and {\sc Perfect Score}.
For {\sc Score}, the input is a deliberation space and a positive integer $\eta$,
and we ask if there is a proposal that is approved
by at least $\eta$ agents in $\cV$; in {\sc Perfect Score},
we are given a deliberation space and 
ask if there is a proposal that is approved by all agents in $\cV$.
These problems model the challenge of finding a good proposal in a centralized way.

As {\sc Perfect Score} is a special case of {\sc Score},
an NP-hardness result for {\sc Perfect Score} implies that {\sc Score}, 
too, is NP-hard, whereas a polynomial-time algorithm for {\sc Score}
can be used to solve {\sc Perfect Score} in polynomial time.

We will also consider another computational problem, which captures
the complexity of the decentralized deliberation process: in $k$-{\sc Compromise},
we are given a deliberative space and a deliberative coalition
structure, and the goal is to find a $k$-compromise transition
from this coalition structure if one exists; note that $k$-{\sc Compromise} 
is a function problem and not a decision problem.

For all three problems, we use prefixes \textit{Euc} and \textit{Hyp} to indicate
whether we consider the variant of the problem for the Euclidean space
or for the hypercube.

Given a deliberation space $I=(\cS,\cV,r,\rho)$,
we denote by $|I|$ the number of bits in the description of $I$.
For $d$-hypercube deliberation spaces, this is simply $O(nd)$
(we need to specify $d$ bits per agent). For Euclidean
deliberation spaces, we assume that the locations of all agents are 
vectors of rational numbers given in binary; similarly, 
when we consider $k$-{\sc Compromise}, we assume that the proposals
supported by the coalitions in the initial deliberative coalition structure
are vectors of rational numbers given in binary.

We will first consider hypercubes, and then
Euclidean spaces.

\subsection{Hypercube Deliberation Spaces}
In hypercube deliberation spaces, even deciding whether
there is a unanimously approved proposal is computationally difficult.

\begin{theorem}\label{thm:hypercubeNPHard}
\text{Hyp}-{\sc Perfect Score} is NP-complete.
\end{theorem}
We prove Theorem~\ref{thm:hypercubeNPHard} by a reduction from {\sc Independent Set}~\cite{gj79}; proof provided in the appendix.

Since {\sc Score} is at least as hard as {\sc Perfect Score}, 
and it is clearly in NP, we obtain the following corollary.

\begin{corollary}
{\it Hyp}-{\sc Score} is NP-complete.
\end{corollary}

Moreover, for hypercube deliberation spaces, we can show that
computing a $k$-compromise is at least as hard as
finding a proposal with a perfect score.

\begin{corollary}
For each $k\ge 2$, there does not exist
a polynomial-time algorithm for \textit{Hyp}-$k$-{\sc Compromise}
unless P=NP.
\end{corollary}
\begin{proof}
Suppose we have a polynomial-time algorithm for \textit{Hyp}-$k$-{\sc Compromise}
for some $k\ge 2$. We will explain how it can be used to solve
\textit{Hyp}-{\sc Perfect Score} in polynomial time.

Given an instance $I$ of \textit{Hyp}-{\sc Perfect Score},
we proceed inductively as follows. For each $i\in [n]$, let
$I_i$ be our instance of \textit{Hyp}-{\sc Perfect Score}
restricted to the first $i$ agents. We will explain how to find
a proposal $x_i$ approved by each of the agents $v_1 \dots, v_i$
if it exists (and output `no' if it does not).
For $I_1$, we can set $x_1=v_1$.
To solve $I_i$, we first solve $I_{i-1}$. If the answer for 
$I_{i-1}$ is `no', then we output `no' as well.
Otherwise, we form an instance of \textit{Hyp}-$k$-{\sc Compromise}
that contains the first $i$ agents; the initial deliberative
coalition structure consists of $(\{v_1, \dots, v_{i-1}\}, x_{i-1})$
and the singleton deliberative coalition $(\{v_i\}, v_i)$.
We can then run the algorithm for \textit{Hyp}-$k$-{\sc Compromise}, $k\ge 2$,
on this instance; it returns a proposal 
if and only if $I_i$ is a yes-instance of {\sc Perfect Score}.
\end{proof}

On the positive side, the problem of finding a popular proposal
becomes easy if the number of agents or the number of dimensions is small:
indeed, we can obtain an FPT algorithm 
with respect to each of these parameters.

\begin{proposition}\label{thm:hypercubeFPT-d}
Given a $d$-hypercube deliberation space, we can compute a popular proposal
in time $O(2^ddn)$. 
\end{proposition}
\begin{proof} 
We can go through the proposals one by one;
it takes $O(d)$ time to decide whether a given agent prefers a given proposal to the status quo, so we can determine the score of each proposal in time $O(dn)$. As there are $2^d-1$
proposals other than the status quo, the bound on the running time follows.
%Edith: for this to work, we need the input to be represented succinctly 
%and we implicitly assume that it is not. So let's not go there...
%
%We note that if $n$ is large compared to $2^d$, then multiple agents may on the same corner-point of the hypercube. If the number of agents at each corner-point of the hypercube is given (each such value uses at most $O(\log n)$ bits), then the number of agents who support a given proposal can be found in $O(d 2^d \log n)$ time. So, overall time complexity is $O(d 2^d \min(n, 2^d \log n) )$.
\end{proof}

\begin{proposition}\label{thm:hypercubeFPT-n}
Given a $d$-hypercube deliberation space, we can compute 
a popular proposal in time
$\textrm{poly}(2^{n2^n}, \log d)$.
\end{proposition}
\begin{proof}
Given an $n$-bit vector $\bb = (b_1, \dots, b_n)\in\{0, 1\}^n$, 
we say that a dimension $j\in [d]$ is of {\em type} $\bb$ 
if for each $i\in [n]$ the $j$-th coordinate of $v_i$ 
is equal to $b_i$. Thus, each dimension belongs to one of the $2^n$ possible types.
For a type $\bb$, let $n_{\bb}$ denote the number of dimensions of type $\bb$. 
We can then represent a proposal $x$ by a sequence of numbers 
$(x_{\bb})_{\bb\in \{0, 1\}^n}$, where $0\le x_{\bb}\le n_{\bb}$:
the value $x_{\bb}$ indicates the number of dimensions of type $\bb$
that are set to $1$ in $x$. Now, for each 
subset of agents $S \subseteq \cV$, we formulate a set of constraints
on the values $(x_{\bb})_{\bb\in \{0, 1\}^n}$ which ensure 
that $x$ is approved exactly by the agents in $S$.

\begin{enumerate}
    \item For each agent $i$ in $S$,
    \begin{multline*}
        \sum_{\bb : b_i = 0} x_{\bb} + \sum_{\bb : b_i = 1} (n_{\bb} - x_{\bb}) < \sum_{\bb : b_i = 1} n_{\bb} \\
        \Longleftrightarrow \sum_{\bb : b_i = 0} x_{\bb} - \sum_{\bb : b_i = 1} x_{\bb} \le -1.
    \end{multline*}
    In the constraint above, $\sum_{\bb : b_i = 0} x_{\bb}$ counts the number of dimensions where the agent has value $0$, but the proposal has value $1$, while $\sum_{\bb : b_i = 1} (n_{\bb} - x_{\bb})$ counts the number of dimensions where the agent has value $1$ but the proposal has value $0$; overall, it measures the distance of the agent from the proposal. $\sum_{\bb : b_i = 1} n_{\bb}$ measures the distance of the agent from the status quo.
    \item For each agent $i$ not in $S$,
    \begin{multline*}
        \sum_{\bb : b_i = 0} x_{\bb} + \sum_{\bb : b_i = 1} (n_{\bb} - x_{\bb}) \ge \sum_{\bb : b_i = 1} n_{\bb} \\
        \Longleftrightarrow \sum_{\bb : b_i = 0} x_{\bb} - \sum_{\bb : b_i = 1} x_{\bb} \ge 0.
    \end{multline*}
    \item Feasibility constraints: for each $\bb \in \{0,1\}^n$, $0 \le x_{\bb} \le n_{\bb}.$
\end{enumerate}
The constraints above form an ILP with $2^n$ variables and $n + 2^n$ constraints. Lenstra’s algorithm (and subsequent improvements of it) can solve an ILP in time exponential in the number of variables, but linear in the number of bits required to represent the problem (constraints). As $n_{\bb}\le d$ and can be represented in $O(\log d)$ bits
for each $\bb\in\{0, 1\}^n$, we obtain a time complexity of $\textrm{poly}(2^{n2^n}, \log d)$. We can find the optimal proposal by searching over all $2^n$ subsets of agents $S \subseteq \cV$. The overall running time is then 
$2^n\cdot \textrm{poly}(2^{n2^n}, \log d) = \textrm{poly}(2^{n2^n}, \log d)$. 
\end{proof}

%%%%%%%%%%%%%%%%%%%%%%%%%%%%%%%%%%%%%%%%%%%%%%%%%%%%%%%%%%%%%%%%%%%%%%%
\subsection{Euclidean Deliberation Spaces}
Recall that, in a Euclidean space, 
an agent $v$ approves approves a proposal $p$ if and only if
\begin{equation}\label{eq:prf:euclidFPT}
    \rho(v,p) < \rho(v,0) \Longleftrightarrow || p ||^2 < 2 \langle v , p \rangle.
\end{equation}
That is, a given proposal $p$ splits $\bR^d$ into two half-spaces; the half-spaces are divided by the hyper-plane orthogonal to and bisecting the line segment joining $p$ to the origin. We will use this correspondence between proposals and half-spaces 
throughout this section.

We first observe that, in contrast to hypercube spaces, 
for Euclidean deliberation spaces
the {\sc Perfect Score} problem is easy.

\begin{proposition}\label{prop:euc-score-n}
\textit{Euc}-{\sc Perfect Score} is polynomial-time solvable.
\end{proposition}
\begin{proof}
It suffices to check whether there exists a hyperplane $\mathcal H$ 
passing through $r$ such that the entire set $\cV$ lies
on the same side of this hyperplane.
This is equivalent to checking whether there exists an $x \in \bR^d$ that satisfies the following linear program: $\langle v , x \rangle > 0, \forall v \in \cV$. 
Once we find such an $x$, we can choose the 
proposal to be $p = \epsilon x / ||x|| $ for sufficiently small $\epsilon > 0$. 
\end{proof}

However, finding a proposal that enjoys a specific degree of support
turns out to be computationally challenging, 

\begin{theorem}\label{thm:euclidNPHard}
\textit{Euc}-{\sc Score} is NP-complete.
\end{theorem}
We prove Theorem~\ref{thm:euclidNPHard} by a reduction from 3-SAT; proof provided in the appendix.

Just like for hypercubes, {\textit Euc}-{\sc Score}
is fixed-parameter tractable with respect to the number of agents.

\begin{proposition}\label{thm:euclidFPT-n}
{\textit Euc}-{\sc Score} can be solved in time $2^n\cdot\mathrm{poly}(|I|)$.
\end{proposition}
\begin{proof}
We reuse the argument in the proof of Proposition~\ref{prop:euc-score-n},
but apply it to every subset of agents.
That is, for every subset $S \subseteq \cV$ of agents, we check whether there exists a hyper-plane passing through the origin that has all the agents in $S$ strictly on one side of the hyper-plane, by solving a linear program. 
Altogether, we have to solve $2^n$ linear programs, so the bound on the running time holds.
\end{proof}

However, for the Euclidean case, we were unable to obtain an FPT algorithm
with respect to the number of dimensions. On the positive side, we can place our problem
in the complexity class XP with respect to this parameter, i.e., show that
it can be solved in polynomial time for the practically important case
where the dimension of the underlying Euclidean space is small.

\begin{proposition}\label{thm:euclidFPT-d}
{\textit Euc}-{\sc Score} can be solved
in time polynomial in $n^{d+1}$ and $|I|$.
\end{proposition}
\begin{proof}
As we observed earlier, 
a proposal divides $\bR^d$ into two half-spaces. The set of half-spaces in $\bR^d$ has a VC-dimension of $d+1$~\cite{kearns1994introduction}, and therefore, the set of proposals also has a VC-dimension of at most $d+1$. Given a set $\cV$ of $n$ agents, 
let $\cS$ be the following set of subsets of agents:
\begin{multline*}
    \cS = \{ S \subseteq \cV \mid \exists x \in \cX \text{ s.t. $|| x ||^2 < 2 \langle v , x \rangle, \forall v \in S$,} \\
    \text{and $|| x ||^2 
    \ge 2 \langle v , x \rangle, \forall v \notin S$} \}.
\end{multline*}
In other words, for every set $S \in \cS$, there exists a proposal that is supported by all the agents in $S$ and none of the agents not in $S$. As the set of proposals has a VC-dimension of $d+1$, the set $\cS$ is of size $O(n^{d+1})$ and the elements of $\cS$ can be enumerated in time polynomial in $n$ when $d$ is fixed.\footnote{See Lemma~3.1 and 3.2 of Chapter 3 of the book by \citet{kearns1994introduction}}. The elements of $\cS$ can be computed inductively, see the proof of Lemma~3.1 in~\citet{kearns1994introduction}. We are also using the fact that given an arbitrary set of agents $S$, we can efficiently compute a proposal that is supported by exactly the agents in $S$ by solving an LP, as shown in Proposition~\ref{prop:euc-score-n}. So, we can find the largest $S$ in $\cS$ and the corresponding proposal efficiently.
\end{proof}

To conclude this section, we consider the complexity of $2$-{\sc Compromise}
in $d$-Euclidean deliberation spaces. Recall that \citet{elkind2021arxiv}
prove (Theorem 3) that a sequence of $2$-compromises is guaranteed to converge
to a successful deliberative coalition structure. Their argument proceeds
by showing that whenever a deliberative coalition structure is not successful, 
then one of the following conditions holds: (i) two existing coalitions can merge; (ii) there exists an agent that can join a maximum-size coalition (the new coalition may need to choose a proposal that differs from the proposal originally supported by the maximum-size coalition), or (iii) the deliberative coalition structure consists of two coalitions (and hence there is a $2$-compromise transition).
Note that the transitions in (ii) and (iii) increase the size of the largest coalition 
and (i) does not decrease it, whereas (i) decreases the number of coalitions.
Consequently, one can reach a successful outcome in $O(n^2)$ steps
by verifying conditions (i)--(iii) and performing the respective transition
when one of them holds. Now, one can efficiently verify whether condition (i) or (ii) holds, and compute the outcome of the respective transition; if this was also
true for (iii), we would be able to solve \textit{Euc}-{\sc Score} in polynomial time, 
in contradiction to Theorem~\ref{thm:euclidNPHard} (assuming P$\neq$NP). We obtain the following corollary.

\begin{corollary}
For each $k\ge 2$, 
there is no polynomial-time algorithm 
for \textit{Euc}-$k$-{\sc Compromise} unless P$\neq$NP.
\end{corollary}
% \abheek{TODO: For constant $d$ and arbitrary $n$, what can we say about the complexity of the problem? This can be answered if we get the answer of: Can we list all subsets of a set $S$ in $\bR^d$ \textit{realized} by half-spaces in time polynomial in $|S|$?}

\section{Number of Transitions}\label{sec:length}
In this section, we go back to viewing deliberation as a decentralized process, 
and ask whether all sequences of $k$-compromise transitions terminate after a number of steps that is polynomial in $n$. Recall that, 
in Euclidean spaces, if the agents are told to perform $2$-compromise transitions in a specific easy-to-compute order, then deliberation terminates in $n^2+1$ steps~\cite{elkind2021arxiv}  (see also the exposition at the end of Section~\ref{sec:complex}).
However, if the agents can choose the order of transitions arbitrarily, 
the only known upper bound (which applies to all deliberation spaces)
is $n^n$, proved using a lexicographic potential function. 
In the next theorem, we put forward a different potential function, 
which enables us to improve the upper bound to $2^n$. 

\begin{theorem}\label{thm:compromiseUpperBd}
The number of transitions in a $2$-deliberation is at most $2^n$. 
This can be shown using the following potential function: \begin{equation}\label{eq:potential2Compromise}
    \phi(\bD) = - |\bD| + \sum_{(C,x) \in \bD} 2^{|C|}.
\end{equation}
\end{theorem}
\begin{proof} %[Theorem~\ref{thm:compromiseUpperBd}]
Consider a deliberative coalition structure $\bD$. 
% As given in the theorem statement, the potential function $\phi$ is defined as:
% \begin{equation*}
%     \phi(\cC) = - |\cC| + \sum_{C \in \cC} 2^{|C|}.
% \end{equation*}
Note that $\bigcup_{(C,x) \in \bD}C = \cV$.
From now on, to simplify notation, 
we will identify a deliberative coalition $(C ,x)$ with its set of agents 
$C$, i.e., we will speak of a coalition $C$ in $\bD$.

Suppose a $2$-compromise transition occurs, where two coalitions $A$ and $B$ of sizes $a$ and $b$ in $\bD$ compromise to form coalitions $C$, $D$, and $E$ of sizes $c$, $d$ and $e$ in $\bD'$, where $D\subsetneq A$, $E\subsetneq B$. We can assume without loss of generality that $a \le b < c$. Note that $D$ and $E$ may be empty, in which case they do not appear in the new coalition structure $\bD'$. 
We have $a + b = c + d + e$, 
and the value of $|\bD'| - |\bD|$ may be either $-1$, $0$, or $1$, depending upon whether $D$ and/or $E$ are empty. 
The change in potential is 
\[
    \phi(\bD') - \phi(\bD) = 2^c - 2^a - 2^b + 1 + (2^d - 1) + (2^e - 1).
\]
Indeed, if $D$ is non-empty then it contributes $2^d$ to $\sum_{C \in \bD'} 2^{|C|}$ and $-1$ to the $-|\bD'|$ portion of $\phi(\bD')$. On the other hand, if $D$ is empty, then it makes neither of these contributions to $\phi(\bD')$, but also $2^d-1=0$. The same argument applies to $E$. As $2^d - 1$ and $2^e - 1$ are always non-negative, we obtain
\[ 
    \phi(\bD') - \phi(\bD) \ge 2^c - 2^a - 2^b + 1 \ge 2\cdot 2^b - 2^b - 2^b + 1 = 1, 
\]
as $c > b \ge a$. Hence, any $2$-compromise transition increases the potential by at least $1$. On the other hand, we have $n \le \phi(\bD) \le 2^n - 1$ for any $\bD$.

% We repeatedly use the fact that $(2^d - 1)$ and $(2^d - 1)$ are always non-negative; on a case by case basis, the change in potential is:
% \begin{itemize}
%     \item If $c > b+1$, then $\phi(\cC') - \phi(\cC) \ge 2^c - 2^a - 2^b + 1 \ge 4*2^b - 2^a - 2^b + 1 \ge 2*2^b+1 \ge 1$. 
%     \item If $c = b+1$ and $a = 1$, then $\phi(\cC') - \phi(\cC) = 2*2^b - 2^a - 2^b + 1 \ge 1$.
%     \item If $c = b+1$ and $a \ge 2$, then $\phi(\cC') - \phi(\cC) \ge 2*2^b + 2^d - 2^a - 2^b - 1 \ge 2^d - 1 \ge 1$.
% \end{itemize}
\end{proof}
Observe that Theorem~\ref{thm:compromiseUpperBd} is independent of the deliberation space and is a property of $2$-compromise transitions. A similar result can be proven for $k$-compromise transitions: they converge in at most $k^n$ steps.

We now focus on proving a lower bound on the convergence of $2$-compromise transitions. 
We first prove a lemma that applies to any deliberation space that satisfies a certain property. We then use this lemma to construct a family of examples
for hypercube deliberation spaces, and for Euclidean deliberation spaces.

\begin{lemma}\label{thm:compromiseSlowConverg}
Suppose a deliberation space with the set of proposals $\cX$ and the set of agents $\cV$ satisfies the following property: for every $C \subseteq \cV$ there exists a proposal $p \in \cX$ s.t. all agents in $C$ approve $p$ and none of the agents not in $C$ approve $p$. Then, a deliberation may take $\Omega(2^{\sqrt{n}/2})$ $2$-compromise transitions.
\end{lemma}
\begin{proof}
Fix a coalition structure $\bD$.
The property of the deliberation space formulated in the theorem statement implies that
for every pair of deliberative coalitions $(A, x), (B, y)\in\bD$ and every $C\subseteq A\cup B$ we can find a proposal $z$ approved by agents in $C$ (and no other agents).
Thus, in what follows, we can reason in terms of sets of agents rather than deliberative coalitions. Consequently, to simplify notation, 
we will write $C\in\bD$ instead of $(C, x)\in \bD$.

We now prove that if the agents end up executing the following types of $2$-compromise transitions, then they will take an exponential time to converge, starting from the coalition structure where all the agents are in singleton coalitions.
\begin{enumerate}
    \item \textbf{Type 1.} If there is a pair $C,C' \in \bD$ such that $|C| = |C'|$, then make the following transition (expressed in terms of coalition sizes):
    \begin{equation*} %\label{eq:move1}
        (a) + (a) \longrightarrow (a+1) + \left\lfloor\frac{a-1}{2}\right\rfloor + \left\lceil\frac{a-1}{2}\right\rceil,
    \end{equation*}
    where $a = |C| = |C'|$. The change in the potential function (as defined in \eqref{eq:potential2Compromise}) for such a transition is
    \begin{multline}\label{eq:potentialChangeType1}
        \Delta \phi \le 2^{a+1} + 2^{\floor{\frac{a-1}{2}}} + 2^{\ceil{\frac{a-1}{2}}} - 2 \cdot 2^{a} \\
        = 2^{\floor{\frac{a-1}{2}}} + 2^{\ceil{\frac{a-1}{2}}} \le \frac{3}{2} \cdot 2^{a/2}. % & \text{ if $a$ odd, then $= \sqrt{2} 2^{a/2} \le \frac{3}{2} 2^{a/2} $ }
    \end{multline}
    (Note that $\floor{\frac{a-1}{2}}$ and $\ceil{\frac{a-1}{2}}$ may be zero, which implies that the number of coalitions decreases, and therefore, may contribute $1$ to $\Delta \phi$ due to change in $|\bD|$. 
    But, in that case, $2^{\floor{\frac{a-1}{2}}}$ would not be in $\Delta \phi$; compensates.)
    % But, if $\floor{\frac{a-1}{2}}$ is $0$, then $2^{\floor{\frac{a-1}{2}}}$ would not be present in $\Delta \phi$, so $2^{\floor{\frac{a-1}{2}}}$ compensates for the $1$. Similar argument applies to $\ceil{\frac{a-1}{2}}$ and $2^{\ceil{\frac{a-1}{2}}}$.)
    \item \textbf{Type 2.} If there are no Type 1 transitions available, then select two smallest coalitions $C,C' \in \bD$ and make the following transition:
    \begin{equation}\label{eq:move2}
        (a) + (b) \longrightarrow (b+1) + \left\lfloor\frac{a-1}{2}\right\rfloor + \left\lceil\frac{a-1}{2}\right\rceil,
    \end{equation}
    where $a = |C|$, $b = |C'|$ and $a \le b$. 
\end{enumerate}

Now, if $\max_{C \in \bD} |C| \le \sqrt{n}$, then there must be at least one pair of coalitions of the same size. If not, then 
\[ \sum_{i=1}^{\sqrt{n}} i = \sqrt{n} (\sqrt{n}+1)/2 = n/2 + \sqrt{n}/2 < n, \]
for large enough $n$, we get a contradiction. So, only Type 1 transitions are made until $\max_{C \in \bD} |C|$ exceeds $\sqrt{n}$.

From \eqref{eq:potentialChangeType1}, until $\max_{C \in \bD} |C| \le \sqrt{n}$ holds, we know that the change in potential for Type 1 transitions is bounded by
\[ \Delta \phi \le \frac{3}{2} \cdot 2^{a/2} \le \frac{3}{2}\cdot 2^{\sqrt{n}/2}.\]
We also know that if $\max_{C \in \bD} |C|$ goes above $\sqrt{n}$ then the value of $\phi$ must exceed $2^{\sqrt{n}} - n$, and the initial value of $\phi$ is $n$. Hence, the number of transitions must be at least
\[ \frac{2}{3}\cdot \frac{2^{\sqrt{n}} - 2n}{2^{\sqrt{n}/2}} = \frac{2}{3}\cdot \left(2^{\sqrt{n}/2} - \frac{2n}{2^{\sqrt{n}/2}}\right). \]
\end{proof}

% We will now apply Lemma~\ref{thm:compromiseSlowConverg} to hypercube deliberation spaces, and then to Euclidean deliberation spaces.
%%%%%%%%%%%%%%%%%%%%%%%%%%%%%%%%%%%%%%%%%%%%%%%%%%%%%%%%%%%%%

\subsection{Hypercube Deliberation Spaces}
In the next theorem, we apply Lemma~\ref{thm:compromiseSlowConverg} to hypercube deliberation spaces by constructing a family of the deliberation spaces that satisfies the required conditions for the lemma.
\begin{theorem}\label{thm:hypslow}
There exists a family of hypercube deliberation spaces where a $2$-deliberation may take $\Omega(2^{\sqrt{n}/2})$ $2$-compromise transitions.
\end{theorem}
\begin{proof}
For every even $d\ge 2$, we will construct an instance with $n = d/2$ agents. All these agents have $0$s in their first $d/2$ dimensions and $1$s in all but one (which is different for each agent) of the last $d/2$ dimensions. 
In particular, the $i$-th agent $v_i = (v_{i,j})_{j \in [d]}$ has
\[
    v_{i,j} = \begin{cases} 0, &\text{if $j \le d/2$ or $j-1-d/2 = i$} \\
    1, &\text{otherwise} \end{cases}.
\]

Consider a set $S$ of $m < d/2$ agents. They agree on $d/2 - m$ 1s. Now consider a proposal that has $d/2 - m - 1$ 1s somewhere in the first $d/2$ dimensions and $d/2 - m$ 1s in the dimensions where the agents in $S$ agree on, and 0s everywhere else. Each agent in $S$ is at a distance of $(d/2-m-1)+ ((d/2-1)-(d/2-m)) = d/2-2$ from this proposal, and at a distance of $d/2-1$ from the origin, so they approve the proposal. However, an agent not in $S$ is at a distance of at least $(d/2-m-1) + ((d/2-1)-(d/2-m-2)) = d/2$ from the proposal, so they do not approve the proposal.

So, for any subset of agents $S$, there is a proposal that is approved exactly by the agents in $S$. Applying Lemma~\ref{thm:compromiseSlowConverg}, we get the desired result.
\end{proof}

%%%%%%%%%%%%%%%%%%%%%%%%%%%%%%%%%%%%%%%%%%%%%%%%%%%%%%%%%%%%%%%%%%%%%

\subsection{Euclidean Deliberation Spaces}
As we did for hypercube deliberation spaces, we construct a family of Euclidean deliberation spaces that satisfies the required conditions for Lemma~\ref{thm:compromiseSlowConverg}.

\begin{theorem}\label{thm:euclidSlowConverg}
There exists a family of Euclidean deliberation spaces where a $2$-deliberation may take $\Omega(2^{\sqrt{n}/2})$ $2$-compromise transitions.
\end{theorem}
% The proof it provided in the Appendix; it works by constructing an Euclidean deliberation space that satisfies the condition required for Lemma~\ref{thm:compromiseSlowConverg}.
\begin{proof} %[Theorem~\ref{thm:euclidSlowConverg}]
For each $d\ge 2$, we will construct an instance with $d$ agents. Let the agent $v_i \in \cV$ be positioned on the $i$-th axis at a distance of $1$ from the origin, i.e., 
agent $v_1$ is located at $(1,0, \ldots, 0)$, agent $v_2$ at $(0,1,0,\ldots,0)$, and so on. 
% agent $v_i = (v_{i,j})_{j \in [d]}$ has
% \[
%     v_{i,j} = \begin{cases} 1, &\text{if $j = i$} \\
%     0, &\text{otherwise} \end{cases}.
% \]

For every $S \subseteq \cV$, let the point $x^S$ be defined as
\[
x^S_i = \begin{cases} 1/|S|, & \text{ if $v_i \in S$} \\
0, & \text{ otherwise}
\end{cases}.
\]
Observe that the distance of agent $v_i \in S$ from $x^S$ is $\rho(v_i,x^S) = (1-1/|S|)^2 + (1/|S|)^2(|S|-1) = 1 - 1/|S| < 1$, so agent $v_i$ prefers $x^S$ to the status quo. But, for an agent $v_i \notin S$, the distance of $v_i$ from $x^S$ is $1 + |S| (1/|S|)^2 = 1 + 1/|S| > 1$. So, for any subset of agents $S$, there is a proposal $x^S$ that is supported exactly by the agents in $S$. Applying Lemma~\ref{thm:compromiseSlowConverg}, we obtain the desired result.
\end{proof}

%%%%%%%%%%%%%%%%%%%%%%%%%%%%%%%%%%%%%%%%%%%%%%%%%%%%%%%%%%%%%
\section{Beyond Two-Way Compromises}\label{sec:beyond}
\citet{elkind2021arxiv} showed that, while in Euclidean deliberation spaces $2$-com\-pro\-mise transitions guarantee successful deliberation, 
in hypercube deliberation spaces this is not the case. Specifically, they proved the following result:
\begin{theorem}\cite{elkind2021arxiv}\label{thm:kCompromiseOld}
There are $d$-hypercube deliberation spaces where $d$-compromise transitions are necessary for a successful deliberation; on the other hand, in every $d$-hypercube deliberation space, $(2^{d-1} + (d+1)/2)$-compromise transitions are sufficient for a successful deliberation.
\end{theorem}
Theorem~\ref{thm:kCompromiseOld} leaves a big gap between the lower bound of $d$ and the upper bound of $2^{d-1} + (d+1)/2$. We tighten this bound by proving a lower bound of $2^{\Theta(d)}$.

\begin{theorem}\label{thm:expCompHyp}
There are $d$-hypercube deliberation spaces where $2^{\Theta(d)}$-compromise transitions are necessary for a successful deliberation.
\end{theorem}
The proof of Theorem~\ref{thm:expCompHyp} is given in the appendix. To prove that $k$-compromise transitions are necessary for a successful deliberation, 
we need to describe a deliberation space and a coalition structure, where:
\begin{enumerate}
    \item The coalition structure is sub-optimal, i.e., there exists a coalition structure with a larger coalition. One way to prove this is to show that an $\ell$-compromise transition, $\ell \ge k$, from this coalition structure leads to a strictly larger coalition. That is, we need to describe a particular proposal and a particular set of agents, and argue that these agents support the new proposal, and the new coalition is larger than the current coalitions of all these agents.
    \item Any $\ell$-compromise, for $\ell < k$, does not lead to a strictly larger coalition. First, there must be at least $k$ coalitions in the current coalition structure. Second, we need to prove that for any proposal in the deliberation space, any set of agents that supports this proposal and is contained in fewer than $k$ current coalitions is at most as large as the current coalition of one of the members of this set.
\end{enumerate}
If we focus on single-agent transitions, the second point above says that, in the current coalition structure, any agent in a (weakly) smaller coalition should not support the proposal of a (weakly) larger coalition. So, the proposals of all the coalitions should be different and should not be supported by any agent from an equal or smaller coalition. $k$-compromise transitions include single-agent transitions and other more complex transitions, and the construction should address all of them. %, which is quite non-trivial.
\citet{elkind2021arxiv} constructed such an example for $k = d$, we do this for 
\[
    k = \binom{(d-1)/9}{(d-1)/27}-1 \ge 3^{(d-1)/27}-1 = 2^{\Theta(d)}.
\]
% The construction is quite non-trivial, and we believe that the readers may find it technically interesting (given in the appendix).
%\input{euclid}
%\input{hypercube}
\section{Conclusions and Future Work}
We have provided an in-depth investigation of the complexity of deliberation
in two models proposed by \citet{elkind2021arxiv}, 
answering several open questions formulated in that paper.
Our results are mostly negative: in both models we have considered, identifying a successful proposal is hard even for a centralized algorithm, and agents will find it computationally challenging to discover feasible transitions from the status quo. Moreover, a completely
decentralized deliberation procedure, in which groups of agents are free to execute compromise transitions in any order, may take a very long time to converge. Finally, while the Euclidean deliberation spaces have the attractive feature that a successful deliberation is possible even if each transition only involves agents from two coalitions, 
in hypercube spaces we may need transitions that involve exponentially many coalitions, negating the benefits of a decentralized process.

Nevertheless, we do not feel that these negative results mean that we should give up on this model of deliberation. Rather, it would be interesting to identify `islands of tractability', i.e., additional conditions that make this model tractable, both in terms of computational complexity and in terms of the length of the deliberation process and the number of coalitions involved in each transition; our FPT and XP results are a step in that direction. It would also be interesting to complement our theoretical findings with empirical work, checking if natural heuristics enable the agents to quickly converge to good (even if not necessarily optimal) outcomes.

There are several questions concerning the complexity of deliberative coalition formation that are left open by our work. For instance, while Theorem~\ref{thm:euclidSlowConverg} shows that convergence may be slow if the number of dimensions scales with the number of agents, it is not clear if this remains true if the number of dimensions is a fixed constant. Further, throughout the paper, we assume that the space of feasible proposals is the entire metric space. A more general approach is to assume that it is a proper subset of the metric space: this subset can be described implicitly by constraints or, in case it is finite, listed explicitly as part of the input. Of course, the computational complexity questions become trivial if the feasible proposals are listed explicitly, but it is not clear if we can bound the length of deliberation as a polynomial function of the number of proposals; on the other hand, the lower bound arguments of Theorems~\ref{thm:hypslow} and~\ref{thm:euclidSlowConverg} no longer apply. Finally, $k$-compromise transitions, 
as defined by \citet{elkind2021arxiv} can be viewed as better responses in the respective game; it would also be interesting to explore the speed of convergence of {\em best} response dynamics, where a negotiation among $k$ coalitions always results in the largest possible coalition that can be formed by their members.

\bibliography{ref}

\cleardoublepage 
\appendix
% \section*{Complexity of Deliberative Coalition Formation \\ (Paper 4470) \\ Technical Appendix}
\appendix

\section{Omitted Proofs}
\begin{proof}[Proof of Theorem~\ref{thm:hypercubeNPHard}]
Note that {\sc Perfect Score} is in NP, because, given a proposal, we can check whether each agent approves it. To prove that this problem is NP-hard, we give a reduction from {\sc Independent Set}~\cite{gj79}.

It will be convenient to identify proposals with bit vectors in $\{0, 1\}^d$,  
$x \in \cX = \{0,1\}^d$, and agents with subsets of $[d]$, $V \subseteq [d]$. 

We start by giving a characterization for the {\sc Perfect Score} problem that will be useful for our proof. The distance of agent $V$ from the origin (the status quo) is $|V|$. Now, agent $V$ prefers a proposal $x \in \{0,1\}^d$ to the status quo if and only if
\begin{multline*}
    \sum_{i \in V} (1-x_i) + \sum_{i \in [d] \setminus V} x_i < |V| \\
    \Longleftrightarrow 1 + \sum_{i \in [d] \setminus V} x_i \le \sum_{i \in V} x_i.
\end{multline*}
We will refer to the above inequality as the {\em characteristic inequality} of agent $V$. It has a special structure: all the $d$ coordinates of $x$ are present either on the left-hand side or on the right-hand side. There is a one-to-one correspondence between inequalities of this form
and subsets of $[d]$.

We are now ready to present our construction. 
Consider an instance of {\sc Independent Set}:
there are $m$ vertices and a set $E$ of edges, and the goal is to decide whether there exists an independent set of size $\kappa$. Let $x_i$ denote the binary variable indicating whether the $i$-th vertex is included in the independent set. The constraints for {\sc Independent Set} are: $x_i + x_j \le 1$ for each edge $(i,j) \in E$; and $\sum_{i \in [m]} x_i \ge \kappa$.
We create a hypercube deliberation space 
with $d = 2 m + 2 \kappa - 1$ dimensions and $O(m)$ agents
so that the input instance of {\sc Independent Set} is a yes-instance
if and only if there exists a proposal that is approved by all agents.

We introduce two dimensions for each binary variable $x_i$ in the {\sc Independent Set} instance, 
and let the coordinates of the proposal along those dimensions be denoted by $x_i$ and $x_i'$. 
Along the $2\kappa-1$ additional dimensions, let the coordinates be denoted by $\alpha_0$, $\alpha_i, \alpha_i'$, $i \in [\kappa-1]$, and denote their tuple by $A$.

We impose constraints on the proposal by creating agents as described below. The first set of agents is not related to the constraints of {\sc Independent Set}.
\begin{enumerate}
    \item For each $\alpha \in A$, we set $\alpha = 1$ by adding $2$ agents as follows. Pick a set $B$ of $\kappa-1$ variables in $A \setminus \{\alpha\}$, and add agents that correspond to the following inequalities:
    \begin{align}
        \sum_{i \in [m]} x_i + \sum_{\alpha' \in B} \alpha' + 1 &\le \alpha + \sum_{i \in [m]} x_i' + \sum_{\alpha' \in A \setminus B} \alpha', \label{eq1}\\
        \sum_{i \in [m]} x_i' + \sum_{\alpha' \in A \setminus B} \alpha' + 1 &\le \alpha + \sum_{i \in [m]} x_i + \sum_{\alpha' \in B} \alpha'.\label{eq2}
    \end{align}
    Summing up constraints~\eqref{eq1} and~\eqref{eq2}, we obtain $\alpha \ge 1$ and hence $\alpha = 1$. From now on, to simplify notation, in each constraint we will use an odd number of variables from $A$: the remaining variables in $A$ (an even number of them) can be equally distributed to the two sides of the inequality, and cancel out as they are all $1$. 
    \item For each $i\in [m]$, we set $x_i = x_i'$ by adding $4$ agents as follows:
    \begin{align}
        x_i + \sum_{j \in [m] \setminus \{i\}} x_j + 1 &\le \alpha_0 + x_i' + \sum_{j \in [m] \setminus \{i\}} x_j'\label{eq3}\\
        x_i + \sum_{j \in [m] \setminus \{i\}} x_j' + 1 &\le \alpha_0 + x_i' + \sum_{j \in [m] \setminus \{i\}} x_j\label{eq4}\\
        x_i' + \sum_{j \in [m] \setminus \{i\}} x_j + 1 &\le \alpha_0 + x_i + \sum_{j \in [m] \setminus \{i\}} x_j'\label{eq5}\\
        x_i' + \sum_{j \in [m] \setminus \{i\}} x_j' + 1 &\le \alpha_0 + x_i + \sum_{j \in [m] \setminus \{i\}} x_j.\label{eq6}
    \end{align}
    Using the fact that $\alpha_0 = 1$, from inequalities~\eqref{eq3} and~\eqref{eq4}, we obtain $x_i \le x_i'$, and from~\eqref{eq5} and~\eqref{eq6} we obtain $x_i' \le x_i$. Hence, $x_i = x_i'$.
\end{enumerate}

We now add agents based on the constraints imposed by our instance of {\sc Independent Set}. For an edge $\{i,j\} \in E$, the corresponding {\sc Independent Set} constraint is $x_i + x_j \le 1$, which we encode in {\sc Perfect Score} as:
\begin{multline*}
    x_i + x_j \le 1 \Longleftrightarrow 2(1 + x_i + x_j) \le 2 \cdot 2 \\
    \Longleftrightarrow (1 + \alpha_0) + (x_i + x_j) + (x_i' + x_j') \le \sum_{i = 1, 2} (\alpha_i + \alpha_i') .
\end{multline*}
Similarly, for the constraint $\sum_{i \in [m]} x_i \ge \kappa$ for {\sc Independent Set}, we add:
\begin{multline*}
    2\kappa \le 2\sum_{i \in [m]} x_i \\
    \Longleftrightarrow (1 + \alpha_0) + \sum_{i \in [\kappa-1]} (\alpha_i + \alpha_i') \le \sum_{i \in [m]} (x_i + x_i').
\end{multline*}
This completes our construction.

It is immediate that there is a one-to-one correspondence between an assignment of the variables in a given instance of {\sc Independent Set} to the variables in the constructed instance of {\sc Perfect Score} (because of the constraints $x_i = x_i'$ for $i \in [m]$ and $\alpha = 1$ for $\alpha \in A$). Also, by construction, there is a one-to-one correspondence between satisfying a constraint of the {\sc Independent Set} instance and satisfying the corresponding agent in the {\sc Perfect Score} instance.
\end{proof}

% \subsection{Section~\ref{sec:euclid}}
\begin{proof}[Proof of Theorem~\ref{thm:euclidNPHard}]
Note that the problem is in NP, because, given a proposal and a set of agents (evidence), we can check whether each agent in the set supports the proposal or not and count the total number of agents in the set. For proving the hardness, we give a reduction from 3-SAT.

\textbf{Construction.} Let there be $m$ variables $x_1, \ldots, x_m$ and ${\ell}$ literals in the 3-CNF formula. We construct an instance of the deliberation problem in $\bR^{2m}$, i.e., $d = 2m$. Let us associate two dimensions in the deliberation space to each variable in the CNF, one corresponding to the positive literal $x_i$, and the other to the negative literal $\neg x_i$. For ease of presentation, let us denote the points in $\bR^{2m}$ using a pair-wise notation: every point is a length $m$ vector of pairs, e.g., $v = ((v_1, \bar{v}_1), (v_2, \bar{v}_2), \ldots, (v_m, \bar{v}_m))$. Let there be the following agents:
\begin{enumerate}
    \item \textbf{Type 1} agents. For $i \in [m]$, let there be a very large number $L$ (specified later) of agents located at $a^{(i)}$ defined as:
    \[ a^{(i)}_j = \bar{a}^{(i)}_j = \begin{cases} -1,& \text{if $j = i$}\\
    0,& \text{otherwise.}\end{cases} \]
    In other words, these points have coordinates $(-1,-1)$ for a pair of dimensions among the $m$ pairs and $(0,0)$ for all other pairs.
    
    \item \textbf{Type 2} agents; these correspond to the variables in the 3-CNF formula. For $i \in [m]$, there are a large number $L'$ (specified later) of agents located at each of the points $b^{(i)}$ and $c^{(i)}$ defined as:
    \[ b^{(i)}_j = \begin{cases} 1,& \text{if $j = i$}\\
    0,& \text{otherwise}\end{cases}, \quad \bar{b}^{(i)}_j = 0,\text{ for every $j$, and}  \]
    \[ c^{(i)}_j = 0,\text{ for every $j$}, \quad \bar{c}^{(i)}_j = \begin{cases} 1,& \text{if $j = i$}\\
    0,& \text{otherwise.}\end{cases}  \]
    In other words, these points have coordinates $(1,0)$ or $(0,1)$ for one pair of coordinates among the $m$ pairs and $(0,0)$ for all other pairs.
    
    \item \textbf{Type 3} agents; these correspond to the terms in the 3-CNF formula. For each $3$-term $\tau = (l_1 \vee l_2 \vee l_3)$, there is an agent at a point $e^{(\tau)}$ defined as:
    \begin{align*}
        e^{(\tau)}_j &= \begin{cases} -1,& \text{if $\bar{x}_j \in \{l_1, l_2, l_3\}$}\\
    0,& \text{otherwise}\end{cases},\\
        \bar{e}^{(\tau)}_j &= \begin{cases} -1,& \text{if $x_j \in \{l_1, l_2, l_3\}$}\\
    0,& \text{otherwise.}\end{cases}  
    \end{align*}
    In other words, if there is a literal $x_j$ in the term $\tau$, then $\bar{e}^{(\tau)}_j = -1$, else if there is a literal $\neg x_j$ in the term $\tau$, then $e^{(\tau)}_j = -1$; all other $(2m-3)$ coordinates of $e^{(\tau)}$ that do not correspond to any literal in the term $\tau$ are set to $0$.
\end{enumerate}
From the construction above, we have $\ell$ Type 3 agents. Let $L' = \ell+1$. As there are $L' = \ell+1$ agents located at each of the $2m$ points of Type 2, an optimal proposal will try to get the support of as many Type 2 agents as possible before trying to get the support of any of the Type 3 agents (note that the $L'$ agents located at each of the $2 m$ points either all support a given proposal or all do not). Further, let $L = 2 m L' + 1$; an optimal proposal will try to get the support of as many Type 1 agents as possible before worrying about the Type 2 (or Type 3) agents. Let $\eta = m L + m L' + \ell$. We have the following decision problem: given the deliberation space just constructed, is there a proposal that is supported by at least $\eta$ agents? In the rest of the proof, we show that this decision problem is equivalent to the original 3-SAT problem.

\textbf{$``\Longrightarrow"$} Let the assignment $x = (x_1, \ldots, x_m) \in \{0,1\}^m$ satisfies the 3-CNF formula. We claim that the following proposal $p$ is supported by at least $\eta$ agents:
\[ 
    p_j = \begin{cases} \frac{1}{7 m},& \text{if $x_j = 1$}\\
    \frac{-3}{7 m},& \text{if $x_j = 0$}\end{cases}, \quad 
    \bar{p}_j = \begin{cases} \frac{-3}{7 m},& \text{if $x_j = 1$}\\
    \frac{1}{7 m},& \text{if $x_j = 0$}\end{cases}.
\]
Note that $|| p ||^2 = m (\frac{1}{49 m^2} + \frac{9}{49 m^2}) = \frac{10}{49 m}$. An agent $v$ supports $p$ iff:
\[ || v - p ||^2 < || v - 0 ||^2 \Longleftrightarrow || p ||^2 < 2 \langle v , p \rangle\]
\begin{enumerate}
    \item Type 1 agents. For any $i \in [m]$, we have
    \begin{align*}
        2 \langle a^{(i)} , p \rangle = \frac{2(-1+3)}{7 m} = \frac{4}{7 m} >  \frac{10}{49 m}.
    \end{align*}
    So, the proposal gets the support of all the $m L$ agents of Type 1.
    
    \item Type 2 agents. For $i \in [m]$, if $x_i = 1$, the proposal gets the support of agents at $b^{(i)}$ because
    \[
        2 \langle b^{(i)} , p \rangle = \frac{2(1)}{7 m} > \frac{10}{49 m}.
    \]
    On the other hand, if $x_i = 0$, the proposal gets the support of agents at $c^{(i)}$ because
    \[
        2 \langle c^{(i)} , p \rangle = \frac{2(1)}{7 m} > \frac{10}{49 m}.
    \]
    So, the proposal gets the support of at least $m L'$ agents out of the $2 m L'$ agents of Type 2.
    
    \item Type 3 agents. For every term $\tau$ in the 3-CNF, at least one of the three literals is true for the assignment $x$.
    \begin{itemize}
        \item If three literals are true, we have $2 \langle e^{(\tau)} , p \rangle = \frac{2(-1)(-3-3-3)}{7 m} = \frac{18}{7 m} > \frac{10}{49 m}$.
        \item If two literals are true, we have $2 \langle e^{(\tau)} , p \rangle = \frac{2(-1)(1-3-3)}{7 m} = \frac{10}{7 m} > \frac{10}{49 m}$.
        \item If one literal is true, we have $2 \langle e^{(\tau)} , p \rangle = \frac{2(-1)(1+1-3)}{7 m} = \frac{2}{7 m} > \frac{10}{49 m}$.
    \end{itemize}
    So, the proposal gets the support of all $\ell$ agents of Type 3.
\end{enumerate}
Adding them up, we have shown that the proposal gets the support of at least $m L + m L' + \ell = \eta$ agents.

\textbf{$`` \Longleftarrow "$} We now assume that the 3-CNF formula is unsatisfiable. We shall prove that there does not exist any proposal that can get the support of $\eta$ agents. Let $y = (y_i, \bar{y}_i)_{i \in [m]}$ be a proposal that is supported by the maximum possible number of agents. 

As argued before, an optimal proposal will try to get the support of as many Type 1 agents, then Type 2 agents, and then Type 3 agents, as possible. We have shown that the proposal $p$ specified above gets the support of all Type 1 agents, so any optimal proposal $y$ must also get the support of all Type 1 agents. For every $i \in [m]$, as Type 1 agents located at $a^{(i)}$ support $y$, we have the following inequality:
\begin{equation}\label{prf:eq:hardness1}
    2 \langle a^{(i)} , y \rangle = -2(y_i + \bar{y}_i) > || y ||^2.
\end{equation}

Given inequality~\eqref{prf:eq:hardness1}, we now show that any proposal can only get the support of agents located at either $b^{(i)}$ or $c^{(i)}$, but not both. If $y$ gets the support of $b^{(i)}$ then it satisfies
\begin{equation}\label{prf:eq:hardness2}
    2 \langle b^{(i)} , y \rangle = 2y_i > || y ||^2,
\end{equation}
while if it gets the support of $c^{(i)}$ then it satisfies
\begin{equation}\label{prf:eq:hardness3}
    2 \langle c^{(i)} , y \rangle = 2\bar{y}_i > || y ||^2.
\end{equation}
If we add the three inequalities \eqref{prf:eq:hardness1}, \eqref{prf:eq:hardness2}, and \eqref{prf:eq:hardness3} we get $||y||^2 < 0$, which is a contradiction. So, $y$ gets the support for at most one of $b^{(i)}$ or $c^{(i)}$ for each $i \in [m]$. The proposal $p$ gets the support of the agents located at $m$ out of these $2 m$ points, so $y$ must also get the support of at least $m$ out of these $2m$ points, because $y$ should optimize for Type 2 agents before worrying about Type 3 agents. So, $y$ gets the support of the agents at exactly one of $b^{(i)}$ or $c^{(i)}$ for every $i$.

Let us define an assignment $x$ for the 3-CNF formula as follows:
\[
    x_i = \begin{cases}1,& \text{ if $y$ is supported by agents at $b^{(i)}$ } \\ %(and not at $c^{(i)}$)}\\
    0,& \text{ if $y$ is supported by agents at $c^{(i)}$ } %(and not at $b^{(i)}$)}
    \end{cases}
\]
We know that the 3-CNF formula is not satisfiable for any assignment, including $x$, so there is a term in the CNF formula for which all three literals are $0$ for assignment $x$. If this unsatisfied term is $\tau = (x_i \vee x_j \vee x_t)$ and $x_i = x_j = x_t = 0$, then by the construction of $x$ using $y$, we know that $y$ is supported by agents at $c^{(i)}$, $c^{(j)}$, and $c^{(t)}$; we have the following inequalities:
\begin{multline}\label{prf:eq:hardness4}
    2\bar{y}_i > || y ||^2; \quad 2\bar{y}_j > || y ||^2; \quad 2\bar{y}_t > || y ||^2 \\
    \implies 2(\bar{y}_i + \bar{y}_j + \bar{y}_t) > 3 || y ||^2.
\end{multline}
Also, corresponding to this term $\tau$ of the CNF, we have the Type 3 agent located at a point $e^{(\tau)}$ where
\[ 
    e^{(\tau)}_{\iota} = 0, \quad \bar{e}^{(\tau)}_{\iota} = \begin{cases} -1,& \text{if $\iota \in \{i,j,m\}$}\\
    0,& \text{otherwise}\end{cases}  .
\]
If $y$ is supported by this agent at $e^{(\tau)}$, then we have:
\[ 
    2 \langle e^{(\tau)} , y \rangle = -2(\bar{y}_i + \bar{y}_j + \bar{y}_t) > || y ||^2,
\]
which cannot be true, because if it is true, then we will contradict inequality~\eqref{prf:eq:hardness4}. As the construction is symmetric w.r.t. the positive and negative literals, w.l.o.g., a similar argument applies if the unsatisfied term is of the other seven types: $(x_i \vee x_j \vee \neg x_t), \ldots, (\neg x_i \vee \neg x_j \vee \neg x_t)$. So, there must be at least one agent of Type 3 that does not support $y$, and therefore, the number of agents that support $y$ is strictly less than $m L + m L' + \ell = \eta$.
\end{proof}

\begin{proof}[Proof of Theorem~\ref{thm:expCompHyp}]
We prove the theorem by giving a coalition structure that is sub-optimal, and where a $2^{\Theta(d)}$-compromise is necessary to make progress towards a successful deliberation.

% Notation.
For easier presentation, we use a set-notation to denote agents and proposals, i.e., an agent or a proposal is a subset of $[d]$. A proposal $X \subseteq [d]$ in set-notation is equivalent to a proposal $x \in \{0,1\}^d$ in bit-vector notation, where $i \in X$ iff the $i$-th bit of $x$ is $1$; similarly for agents. We shall overload the notation for the set of agents $\cV$ and the set of proposals $\cX$ for both set and bit-vector notations. The distance between an agent $V \in \cV$ and a proposal $X \in \cX$ can be written as $\rho(V,X) = | X \setminus V | + | V \setminus X | $. As before, w.l.o.g. we assume that the status quo is the empty set (or at the origin in bit-vector notation).
%\footnote{If the status quo is at some other point, taking an exclusive-or of all the agents and the proposals with the status quo, including the status quo, shifts the status quo to the origin without changing the underlying mathematical problem.} 
An agent $V$ supports a proposal $X$ iff
\begin{multline*}
    \rho(V,X) < \rho(V,\phi) \\
    \implies | X \setminus V | + | V \setminus X | < |V| = | V \setminus X | + | V \cap X | \\
    \Longleftrightarrow | X \setminus V | < | V \cap X |  \Longleftrightarrow |X| < 2 | V \cap X |,
\end{multline*}
i.e., $V$ supports $X$ iff $V$ intersects with strictly more than half of $X$.

% \textit{Construction.} 
Let the number of dimensions $d$ be a large positive integer, where $(d-1)$ is odd and is divisible by $27$. We shall use the $d$-th dimension in a special manner, different from the remaining $(d-1)$ dimensions. Let $d' = (d-1)/3$, let $\hat{d} = ((d'/3)-1)/2$ or $2 \hat{d} + 1 = d'/3$. Note that $d'$, $\hat{d}$, $(2\hat{d}+1)$, and $(2\hat{d}+1)/3$ are all integers as per our choice of $d$. Let $k = \binom{d'/3}{(2\hat{d}+1)/3}-1$; we shall prove that a $k$-compromise transition is necessary for successful deliberation.

Let the first $(d-1)$ dimensions be partitioned into $d'$ triplets; further, the $d'$ triplets be partitioned into $d'/3$ triplets of triplets (nonuplets). In other words, we may identify a particular dimension among the $(d-1)$ dimensions as $(i,j,m) \in [d'/3] \times [3] \times [3]$. Throughout the proof, by a triplet we denote the three dimensions given by $(i,j,*)$, and there are $d'$ such triplets; and by a nonuplets we identify the nine dimensions given by $(i,*,*)$, and there are $d'/3$ such nonuplets.

Let $\alpha$ and $\beta$ be two rational numbers strictly between $0$ and $1$, we shall specify their values towards the end of the proof. 

To complete our construction of the deliberation space, we need to specify the agents and their locations. Then, we need to specify the coalition structure that requires a $k$-compromise transitions, i.e., we need to specify all the proposal--coalition pairs in the coalition structure. We specify the coalition structure below, along the way also specifying the set of agents in the deliberation space.

\textbf{Current Proposals (CPs).} In the current coalition structure, let there be $(k+1)$ coalitions formed around $(k+1)$ distinct proposals.  Out of the $d'/3$ nonuplets, each proposal contains exactly $(2 \hat{d} + 1)/3$ distinct nonuplets, which gives us $\binom{d'/3}{(2\hat{d}+1)/3} = k+1$ proposals. Let us call these proposals \textit{current proposals} or CPs. Observe that any two CPs differ by at least one nonuplet, i.e., three triplets, i.e., nine dimensions.

Let us arrange the $(k+1)$ CPs in a sequence $S = (X_1, X_2, \ldots, X_{k+1})$ such that consecutive proposals in the sequence have an empty intersection, i.e., $X_i \cap X_{i+1} = \phi$ for any $i \in [k]$. Such a sequence always exists based on results on the existence of Hamiltonian paths in Kneser graphs~\cite{wikiKneserGraph}. Particularly, a  direct corollary of a result by Chen~\cite{chen2003triangle} says: if $2.62 (2\hat{d}+1) + 1 \le d'$ then such a sequence exists. As $(2\hat{d}+1) = d'/3 \implies 2.62 (2\hat{d}+1) + 1 \le d'$ for large enough $d'$, a sequence $S$ with the required property exists.

\textbf{Agents.} We specify the set of agents in two steps: first, we give the position of an agent, which we call the \textit{type} of the agent (the type determines the proposals that a given agent supports); second, we give the number of agents of each type.
\begin{itemize}
    \item Pick an arbitrary CP $X$ out of the $(k+1)$ CPs. Note that $X$ has $(2\hat{d}+1)$ triplets out of the total $d'$ triplets. Let us construct an agent type $V$ based on the CP $X$ as follows: $V$ has exactly two out of three elements from each of any $(\hat{d}+1)$ triplets out of the $(2\hat{d}+1)$ triplets in $X$ (we call such triplets 2-triplets); one out of three elements from each of the remaining $\hat{d}$ triplets in $X$ (we call such triplets 1-triplets); doesn't have any elements from the $(d'-2\hat{d}-1)$ triplets not in $X$ (we call such triplets 0-triplets); $V$ may or may not include the last element/dimension $d$. Let $\nu(X)$ be the set of types of agents that are constructed using $X$.
    
    Note that $|\nu(X)| = \binom{2\hat{d}+1}{\hat{d}+1} 3^{2\hat{d}+1} 2$, where the factor $\binom{2\hat{d}+1}{\hat{d}+1}$ is for the choice of $(\hat{d}+1)$ 2-triplets out of the $(2\hat{d}+1)$ triplets, $3^{\hat{d}+1}$ is for the choice of not selecting one out of three elements in each of the $(\hat{d}+1)$ 2-triplets, $3^{\hat{d}}$ is for the choice of selecting one out of the three elements in each of the $\hat{d}$ 1-triplets, and the $2$ is for the choice of either selecting or not selecting the element $d$. 
    
    Observe that each $V$ has a size of either $2(\hat{d}+1) + 1(\hat{d}) + 0 = 3\hat{d}+2$ or $2(\hat{d}+1) + 1(\hat{d}) + 1 = 3\hat{d}+3$, depending upon whether the last dimension $d \in V$ or $d \notin V$. On the other hand, a CP $X$ has a size $3(2\hat{d}+1) = 6\hat{d}+3$. If an agent type $V$ has been constructed from the CP $X$, i.e., $V \in \nu(X)$, then $|V \cap X| = 3\hat{d}+2 > (6\hat{d}+3)/2 = |X|$, because all elements in $V$ except $d$ (if $d \in V$) are also in $X$; and therefore, type $V$ agents support $X$. On the other hand, if $V \notin X$, then there are at least three triplets that are in $V$ but not in $X$, and each of these triplets contribute at least one element to $V$, so $|V \cap X| \le 3\hat{d}-1 \le (6\hat{d}+3)/2 = |X|$; and therefore, $V$ doesn't support $X$. So, the agents with types in $\nu(X)$ support $X$ and don't support any other CP $X' \in S \setminus \{X\}$. This implies that the sets $(\nu(X))_{X \in S}$ are all disjoint. There are total $(k+1) |\nu(X)|$ such types, and any agent in the deliberation space is of one of these $(k+1) |\nu(X)|$ types. 
    
    In the current coalition structure, all agents of type $V \in \nu(X)$ are in the coalition that supports the proposal $X$.
    
    \item For ease of presentation, we provide the number of agents of each type as a non-negative rational number; this is without loss of generality and all the results follow when we convert these rational numbers to non-negative integers by multiplying them by the lowest common multiple of all the denominators.
    
    By construction, any given agent type $V$ with the last dimension $d \in V$ has a sibling type $V' = V \setminus \{d\}$, and vice-versa. For a type of agent $V$, let $\eta(V)$ be the number of agents of that type. We set $\eta(V)$ based on: (i) whether $d \in V$ or not; (ii) the location where the CP $X$ to which $V$ is associated with (i.e., the $X$ s.t. $V \in \nu(X)$) lies in the sequence of CPs $S$ defined previously. Let $X_i$ be the $i$-th CP in $S$.
    \begin{equation}
        \eta(V) = \begin{cases} \alpha \beta^{i-1},& \text{ if $d \in V$} \\ 
                                (1-\alpha) \beta^{i-1},& \text{ if $d \notin V$}
                    \end{cases},
    \end{equation}
    where $V \in X_i$, $i \in [k+1]$.
    Notice that the ratio of the number of agents that support $X_{i+1}$ vs the number of agents that support $X_i$ is $=\beta < 1$. So, the size of the coalition at $X_i$ strictly decreases by a multiplicative factor $\beta$ as $i$ increases. And, the largest coalition supports $X_1$, and the smallest $X_{k+1}$.
\end{itemize}
This completes our specification of the agents and the current coalition structure. 

In the rest of the proof, we shall set the values of $\alpha$ and $\beta$ to ensure that there is a proposal $X^*$ such that a $(k+1)$-compromise transition to form a coalition supporting $X^*$ leads to a strictly larger coalition (than all the coalitions in the current coalition structure), while no $(k-1)$-compromise transition leads to a strictly larger coalition (than all the coalitions participating in the compromise transition). This tells us that a $k$-compromise is essential for a successful deliberation, and $k = \binom{d'/3}{(2\hat{d}+1)/3}-1 \ge \left( \frac{d'}{2\hat{d}+1} \right)^{2\hat{d}+1}-1 = 3^{(d-1)/27}-1 = 2^{\Theta(d)}$.

Throughout the proof, we shall come across a few constraints on $\alpha$ and $\beta$. We shall track them and argue that there are rational numbers that satisfy all the constraints. The first two constraints we enforce are:
\begin{equation}\label{prf:expCompHyp:eq:a1}
    0 < \alpha < 1; 0 < \beta < 1.
\end{equation}

Let the proposal that we use for the $(k+1)$-compromise be denoted by $X^*$. $X^*$ has only one element, the last dimension $d$, i.e., $X^* = \{d\}$. All agents of type $V$ such that $d \in V$ support $X^*$. So, $X^*$ gets exactly $\alpha$ fraction of agents from every current coalition, and the total number of agents that support $X^*$ is $\alpha \sum_{i \in [k+1]} \beta^{i-1} = \alpha \frac{1-\beta^{k+1}}{1-\beta}$. We need this number to be strictly bigger than all the coalitions, in particular, we need this number to be bigger than the largest current coalition, which is at $X_1$. So, we need to satisfy the following inequality:
\begin{equation}\label{prf:expCompHyp:eq:a2}
    \alpha \sum_{i \in [l]} \beta^{i-1} = \alpha \frac{1-\beta^{k+1}}{1-\beta} > 1 = |X_1|.
\end{equation}
On the other hand, let us aim to prevent $m$-compromises to $X^*$, where $m \le k$ (the result we ultimately prove if for $m < k$ and not $m \le k$). For any set of $m \in [k]$ current coalitions, if we are able to show that the largest coalition among these $m$ coalitions doesn't have an incentive to do the $m$-compromise, i.e., the coalition formed by the $m$-compromise is equal or smaller than the largest participating coalition, then this $m$-compromise is not valid. Say $X_j$ is the largest participating coalition, the maximum possible size of the new coalition is $\alpha \sum_{i \in [m]} \beta^{j-1+i-1}$, and the size of current coalition is $\beta^{j-1}$. So, we need to satisfy the inequality $\alpha \sum_{i \in [m]} \beta^{j-1+i-1} \le \beta^{j-1}$, which is automatically satisfied if we satisfy the inequality below (as $m \le k$):
\begin{equation}\label{prf:expCompHyp:eq:a3}
    \alpha \sum_{i \in [k]} \beta^{i-1} = \alpha \frac{1-\beta^k}{1-\beta} \le 1.
\end{equation}
The inequalities \eqref{prf:expCompHyp:eq:a2} and \eqref{prf:expCompHyp:eq:a3} are automatically satisfied if we satisfy the equation below (as $\alpha, \beta > 0$, and therefore, $\alpha \beta^k > 0$):
\begin{equation}\label{prf:expCompHyp:eq:a4}
    \alpha \sum_{i \in [l-1]} \beta^{i-1} = \alpha \frac{1-\beta^{k}}{1-\beta} = 1.
\end{equation}
It is easy to check that for any $\alpha \in (0,1)$, we can select a $\beta \in (0,1)$ to satisfy the equation \eqref{prf:expCompHyp:eq:a4} above. (A rational $\beta$ can be found in the neighborhood of the $\beta$ found as the solution of \eqref{prf:expCompHyp:eq:a4}, which still satisfies \eqref{prf:expCompHyp:eq:a2} and \eqref{prf:expCompHyp:eq:a3}.) Moving forward, we shall not put any more constraints on $\beta$, and we will show that the constraints we put on $\alpha$ allow it to be a rational number in $(0,1)$.

Now, we need to exhaustively prove that there is no other proposal that allows a compromise transition of size much smaller than $(k+1)$ to succeed, in particular, we shall prove that there is no $(k-1)$-compromise transition. Let $Y \subset [d]$ be an arbitrary proposal. W.l.o.g. we can assume that $|Y|$ is odd; any agent type $V$ supports $Y$ iff $|V \cap Y| > |Y|/2$; if $|Y|$ is even, then we can create $Y'$ by adding another arbitrary element from $[d]$ that is not in $Y$ (unless $Y = [d]$, which we shall prove is not supported by any agent) and still satisfy the constraint for $V$ because $|V \cap Y'| \ge |V \cap Y| > |Y|/2 \implies |V \cap Y'| > |Y'|/2$. We now prove that there is no $(k-1)$-compromise that leads to a strictly larger coalition at $Y$, on a case by cases basis:

\textbf{Case 1: $|Y| \ge 6\hat{d}+6$.} For any type of agent $V$, $|V \cap Y| \le |V| \le 3\hat{d}+3 = (6\hat{d}+6)/2 \le |Y|/2$. So, no agent supports $Y$.

\textbf{Case 2: $6\hat{d}+4 \le |Y| \le \hat{d}+5$.} For any $V$, we need $|V \cap Y| > |Y|/2 = 3\hat{d}+2 \implies |V| \ge 3\hat{d}+3$. $V$ satisfies the constraint only if $d \in V$, but then $V$ supports $X^*$ too. So, the number of agents that support $X^*$ is at least as many as the number of agents that support $Y$, and therefore, we cannot have a $k$-compromise using the proposal $Y$.
    
\textbf{Case 3: $|Y| \le 6\hat{d}+3$.} We divide this case into two sub-cases based on whether $d \in Y$ or $d \notin Y$.
    \begin{enumerate}
    
    \item $d \notin Y$. Note that for these types of proposals, for a given $V$, whether $V$ supports $Y$ or does not support $Y$, does not depend upon whether $V$ contains $d$ or not. Also, note that $|Y| \le |X|$ for every CP $X \in S$, so either $Y = X$ for some CP $X \in S$, or $Y \cap X \subsetneq X$ for every CP $X \in S$. The case when $Y = X$ for some CP $X \in S$ is trivial, because all agents that support $X$ are already in the coalition with the proposal $X$; the only interesting case is the other one.
    
    Let $Z = Y \cap X$. Pick an arbitrary $X \in S$, we know that $Z \subsetneq X$. We claim that there is at least one type $V \in \nu(X)$ that does not support $Y$. Out of the $(2\hat{d}+1)$ triplets in $X$, let $Z$ contain $a$ of these triplets entirely (we call such triplets 3-triplets), two out of three elements from $b$ triplets (we call such triplets 2-triplets, as before), and one out of three elements from $c$ triplets (we call such triplets 1-triplets, as before). Note that $a \le 2\hat{d}$, otherwise $Z$ would be equal to $X$. Also note that each agent type $V$ has $(\hat{d}+1)$ 2-triplets and $\hat{d}$ 1-triplets. We construct a type $V$ that doesn't support $Z$ as follows: 
    \begin{itemize}
        \item Select as many as possible of the $\hat{d}$ 1-triplets of $V$ from the $a$ 3-triplets of $Z$ (if $a > \hat{d}$, we select 2-triplets of $V$ for the remaining $(a-\hat{d})$ 3-triplets of $Z$). As $a \le 2\hat{d}$ and the number of 1-triplets in $V$ is $\hat{d}$, so the overlap between $V$ and $Z$ in these $a$ triplets is $\le 3a/2$.
        \item Select the remaining $((2\hat{d}+1)-a)$ triplets of $V$ arbitrarily, but we will select the elements inside the triplets to have minimum possible intersection with $Z$. We know that the remaining triplets in $Z$ are either the $b$ 2-triplets or the $c$ 1-triplets. If a given triplet is a 2-triplet of both $V$ and $Z$, then select the elements of $V$ in a manner to have an overlap of only $1$ with the elements of $Z$; while, if a given triplet is a 1-triplet of either of $V$ or $Z$ (or a $0$-triplet of $Z$), then select the elements of $V$ to have no overlap with $Z$. Notice that the total overlap of $V$ with the $b$ 2-triplets of $Z$ is $\le b$, and the total overlap with the $c$ 1-triplets of $Z$ is $0$.
    \end{itemize}
    By construction, the total overlap between $V$ and $Z$, i.e., $|V \cap Z|$ is $\le 3a/2 + b \le (3a + 2b + c)/2 = |Z|/2 \le |Y|/2$. Also, note that $V$ and $Y$ cannot overlap at any element not in $Z$. So, $V$ doesn't support $Y$. 
    % Let $\alpha_{1,X,Y}$ be the fraction of types in $\nu(X)$ that $Y$ supports (it is also the fraction of agents that support $Y$, in the coalition corresponding to $X$), and let $\gamma = \max_{X,Y} \alpha_{1,X,Y}$ where $Y$ is as per this particular case. We have $\alpha_{1,X,Y} < 1 \implies \gamma < 1$.
    Let $\gamma = 1-(1/|\nu(X)|) < 1$. So, $Y$ gets the support of at most $\gamma$ fraction of agents from any coalition.
    Now, if we enforce the following additional constraint on $\alpha$:
    \begin{equation}\label{prf:expCompHyp:eq:a5}
        \alpha \ge \gamma,
    \end{equation}
    then $X^*$ will capture at least as many agents from any current coalition as $Y$, so, $X^*$ is as good as $Y$.
    
    \item $d \in Y$. Let $Y' = Y \cap [d-1] = Y \setminus \{d\}$. As $|Y| \le 6\hat{d}+3$ and $d \in Y$, restricting ourselves to first $(d-1)$ dimensions (the triplets), $|Y'| \le 6\hat{d}+2$. As explained earlier, w.l.o.g. we can assume that $|Y|$ is odd, and therefore, we can assume that $|Y'|$ is even.
    
    $Y$ interacts differently with the $\alpha$ fraction of agents that have the last dimension $d$ in their type and the $(1-\alpha)$ fraction of agents that don't have $d$. For an agent type $V$ with $d \notin V$, $V$ supports $Y$ iff $|V \cap Y'| = |V \cap Y| > |Y|/2 > |Y'|/2$. As $Y' < 3(2\hat{d}+1)$, therefore $|Y' \cap X| < 3(2\hat{d}+1) \implies Y' \cap X \subsetneq X$ for every $X \in S$; this case is similar to the previous sub-case where $d \notin Y$, and we have proved that at most $\gamma < 1$ fraction of agents from any coalition supports $Y$. 
    
    On the other hand, if $d \in V$, $V$ supports $Y$ if $|V \cap Y| > |Y|/2 \Longleftrightarrow |V \cap Y| = |V \cap Y'| + 1 \ge (|Y|-1)/2 + 1 = |Y'|/2 + 1 \Longleftrightarrow |V \cap Y'| \ge |Y'|/2$, as $|Y|$ is odd. So, $V$ supports $Y$ iff $V$ contains at least half of the elements in $Y'$; note that this condition is a slightly relaxed condition than the condition we had for the previous case (the inequality is not strict), and needs a separate analysis, provided below.
    
    Note that $Y' = \phi$ means $Y = X^*$, which we have already considered, so we assume $Y' \neq \phi$. Let $X \in S$ be the CP that has the largest intersection with $Y'$, and let $Z = X \cap Y'$. Let $a$, $b$, and $c$ be the number of 3-triplets, 2-triplets, and 1-triplets in $Z$, respectively. Observe that $a \le 2\hat{d}$ as $|Z| \le |Y'| < 3(2\hat{d}+1)$. Consider the following cases based on the value of $a$:
    \begin{enumerate}
        \item $a < 2\hat{d}$. We claim that there is a type $V$ that doesn't satisfy $|V \cap Y'| \ge |Y'|/2$, i.e., $|V \cap Y'| < |Y'|/2$. We construct such a $V$ in the exact same manner as we did in the previous sub-case (for $d \notin Y$). In particular, we selected as many as possible of the 1-triplets in $V$ from the $a$ 3-triplets of $Z$. As $a < 2\hat{d}$, the overlap between $V$ and $Z$ in these $a$ triplets is $< 3a/2$. For the $b$ 2-triplets and $c$ 1-triplets, we get an overlap of $\le b$ as before. So, $|V \cap Y'| = |V \cap Z| < 3a/2 + b \le |Z|/2 \le |Y'|/2$. So, there is at least one agent type that doesn't satisfy $|V \cap Y'| \ge |Y'|/2$. So, $Y$ gets the support of at most $\gamma$ fraction of agents from any coalition, and enforcing inequality \eqref{prf:expCompHyp:eq:a5}, $X^*$ is as good as $Y$.
        
        % \item $a = 2k$ and $c = 1$ (note $b=0$ and $c$ cannot be $2$ as $a+b+c \le 2k+1$). This case is similar to the previous case 3.(a) (and 3.(b).i.). We bound the overlap between $V$ and $Z$ in the $a$ 3-triplets by $\le 3a/2$, and overlap for the $c=1$ 1-triplet to be $0$. But, we also have $|Z| = 3a+1$ So, $|V \cap Y'| = |V \cap Z| \le 3a/2 < (3a+1)/2 = |Z|/2 \le |Y'|/2$. So, $V$ doesn't support $Y$, and $Y$ gets the support of at most $\gamma$ fraction of agents from any coalition like case 3.(b).i.
        
        % \item $a = 2k$ and $c=0$ ($b$ can be either $0$ or $1$). Notice that $Y' = Z$, because: if $b=1$, then $|Z| = 6k+2 = |Y'|$, and as $Z \subset Y'$, $Z = Y'$; if $b=0$ and $Y' \setminus Z \neq \phi$, then the element in $Y' \setminus Z$ belongs to some other triplet, we can drop the triplet $X \setminus Z$ and add this one to get a larger
        
        \item $a = 2\hat{d}$. This case is slightly tricky, because any $V \in \nu(X)$ automatically satisfies the constraint $|V \cap Y'| \ge |Y'|/2$, and therefore supports $Y$ if $d \in V$. So, $Y$ captures at least as many agents as $X^*$ from the coalition currently supporting $X$; $X^*$ gets the support of exactly $\alpha$ fraction of agents that currently support $X$, but $Y$ gets the support of (at most) $\alpha + (1-\alpha) \gamma$, where $(1-\alpha) \gamma$ comes from the agents of type $V$ with $d \notin V$ (the fraction of agents with $d \notin V$ is $(1-\alpha)$ and from our previous discussion we know that at most $\gamma$ fraction of them can support $Y$ because $Y \cap X \subsetneq X$).
        
        Let $X$ be the $i$-th element in the sequence $S$, i.e., $X = X_i$. We know that $X_i \cap X_{i+1} = \phi$ by the construction of $S$, and $|Y' \cap X_i| \ge 6\hat{d}$ and $|Y'| \le 6\hat{d}+2$, so $|Y' \cap X_{i+1}| \le 2$. So, for any $V \in \nu(X_{i+1})$, $|V \cap Y| \le 1 + |V \cap Y'| \le 3 < (6\hat{d}+1)/2 \le |Y|/2$ for any $\hat{d} \ge 1$. Therefore, no agent in $\nu(X_{i+1})$ will support $Y$.
        
        For a general CP $X' \in S \setminus \{X\}$, $X'$ differs from $X$ by at least three triplets (a nonuplet), by construction. So, at least two out of the $a=(2\hat{d})$ 3-triplets in $Y'$ will not be in $X'$, therefore $|Y' \cap X'| \le 6\hat{d} + 2 - 6 = 6\hat{d} - 4$, which implies that $Y' \cap X'$ has strictly less than $(2\hat{d})$ 3-triplets. This is similar to the previous case (where $a < 2\hat{d}$), and we know that at most $\gamma$ fraction of agents in the coalition that currently supports $X'$ will support $Y$.
        
        Let there be a compromise between $m$ current coalitions that forms a strictly larger coalition at $Y$. Let the proposals of those $m$ coalitions be $(X_{i_1}, X_{i_2}, \ldots, X_{i_m})$, where $(i_j)_{j \in [m]}$ have been arranged in an increasing order, and therefore, the size of coalitions are in decreasing order. Let $I = (i_j)_{j \in [m]}$. The case when $X = X_i \notin \{X_j\}_{j \in I}$ is not very interesting because each of $X_j$ can contribute at most $\gamma$ fraction to the new coalition at $Y$, but can contribute $\alpha \ge \gamma$ fraction to $X^*$. So, let us focus on the case when $X_i \in \{X_j\}_{j \in I}$. The $m$-compromise will not benefit by including $(i+1)$ to $I$ because the coalition at $X_{i+1}$ will not contribute any agents to the new coalition at $Y$, so w.l.o.g. $(i+1) \notin I$. The size of the new coalition at $Y$ is upper bounded by:
        \[ \beta^{i-1} (\alpha + (1-\alpha) \gamma) + \gamma \sum_{j \in I \setminus \{i\}} \beta^{j-1}. \]
        Now, let us consider the following $(m+1)$-compromise to form a new coalition at $X^*$: the $(m+1)$ coalitions that do the compromise transition are located at the proposals $(X_j)_{j \in I \cup \{i+1\}}$. The coalition has size:
        \[ \beta^{i-1} \alpha + \beta^i \alpha + \alpha \sum_{j \in I \setminus \{i\}} \beta^{j-1}. \]
        From equation \eqref{prf:expCompHyp:eq:a4}, $(1-\beta) = \alpha(1-\beta^{k}) \le \alpha$ as $\beta < 1$, which implies $\beta \ge 1-\alpha$. From inequality \eqref{prf:expCompHyp:eq:a5}, $\alpha \ge \gamma$. Combining these two inequalities we get:
        \begin{align*}
            &\beta \alpha \ge (1 - \alpha) \gamma \\
            \implies& \beta^i \alpha \ge \beta^{i-1} (1 - \alpha) \gamma \\
            \implies& \beta^{i-1} \alpha + \beta^i \alpha \\
            \ge& \  \beta^{i-1} \alpha + \beta^{i-1} (1 - \alpha) \gamma = \beta^{i-1} (\alpha + (1 - \alpha) \gamma)\\
            \implies& \beta^{i-1} \alpha + \beta^i \alpha + \alpha \sum_{j \in I \setminus \{i\}} \beta^{j-1} \\
            \ge& \  \beta^{i-1} (\alpha + (1 - \alpha) \gamma) + \gamma \sum_{j \in I \setminus \{i\}} \beta^{j-1}.
        \end{align*}
        The inequality above tells us that the size of the new coalition formed by the $(m+1)$-compromise at $X^*$ is at least as large as the new coalition formed by the $m$-compromise at $Y$. As there are no successful $k$-compromises at $X^*$, so there are no successful $(k-1)$-compromises at $Y$, so any successful compromise will need the contribution from at least $k$ current coalitions.
    \end{enumerate}
    \end{enumerate}
% \end{enumerate}
To complete the proof we need $\alpha$ and $\beta$ that satisfy the constraints \eqref{prf:expCompHyp:eq:a1}, \eqref{prf:expCompHyp:eq:a4}, and \eqref{prf:expCompHyp:eq:a5}, which we can do by setting $\alpha = (1+\gamma)/2$ and $\beta$ as per the solution of equation \eqref{prf:expCompHyp:eq:a4}. (If the $\beta$ found as the solution of \eqref{prf:expCompHyp:eq:a4} is irrational, a rational $\beta$ can be found in its neighborhood that satisfies \eqref{prf:expCompHyp:eq:a2} and \eqref{prf:expCompHyp:eq:a3}, as required.)
\end{proof}

\end{document}